\documentclass[a4paper,11pt]{article}

\usepackage{latexsym}
\usepackage[latin1]{inputenc}
\usepackage{amsmath,amssymb,amsfonts,amsthm}
\usepackage{graphicx}
\usepackage{geometry}
\usepackage[english]{babel}

\newcommand{\beq}{\begin{equation}}
\newcommand{\eeq}{\end{equation}}
\newcommand{\bi}{\begin{itemize}}
\newcommand{\bd}{\begin{description}}
\newcommand{\ei}{\end{itemize}}
\newcommand{\ed}{\end{description}}

\newcommand{\bc}{\begin{center}}
\newcommand{\ec}{\end{center}}

\newtheorem{df}{Definition}[section]
\newtheorem{ass}[df]{Assumption}
\newtheorem{lm}[df]{Lemma}
\newtheorem{prop}[df]{Proposition}
\newtheorem{thm}[df]{Theorem}

\newtheorem{rem}[df]{Remark}

\newcommand{\sg}{\sigma}

\newcommand{\F}{{\mathcal{F}}}

\begin{document}

\title{On the Existence of Martingale Measures in Jump Diffusion Market
Models\thanks{We are grateful to Claudio Fontana, Johannes Ruf and
Ngoc Huy Chau for valuable comments.}}
\author{Jacopo Mancin\\[0.2cm]
Department of Mathematics, LMU University Munich\\mancin.jacopo@gmail.com
 \and Wolfgang J.Runggaldier\\[0.2cm]
Dipartimento di
Matematica Pura ed Applicata, Universit\`a di Padova\\
runggal@math.unipd.it}

\date{}

\maketitle

\begin{abstract} In the context of jump-diffusion market models we
construct examples that satisfy the weaker no-arbitrage condition of
NA1 (NUPBR), but not NFLVR. We show that in these examples the only
candidate for the density process of an equivalent local martingale
measure is a supermartingale that is not a martingale, not even a
local martingale. This candidate is given by the supermartingale
deflator resulting from the inverse of the discounted growth optimal
portfolio. In particular, we consider an example with constraints on
the portfolio that go beyond the standard ones for admissibility.
\vspace{0.2cm}

\vspace{0.1cm}


\end{abstract}

\section{Introduction}

The First Fundamental Theorem of Asset Pricing states the
equivalence of the classical no-arbitrage concept of {\sl
No-Free-Lunch-With-Vanishing-Risk} (NFLVR) and the existence of an
{\sl equivalent $\sg-$martingale measure} (E$\sg$MM) (see
\cite{DS98}). Since in this paper we shall consider only nonnegative
price processes, we may restrict ourselves to the sub-class of {\sl
equivalent local martingale measures} (ELMM), which can be
characterized by their density processes (see \cite{DS94}). In real
markets some forms of arbitrage however exist that are incompatible
with FLVR. This is taken into account in the recent Stochastic
Portfolio Theory (see the survey in \cite{FK}), where the NFLVR
condition is not imposed as a normative assumption and it is shown
that some arbitrage opportunities may arise naturally in a real
market; furthermore, one of the roles of portfolio optimization is
in fact also that of exploiting possible arbitrages. On the other
hand the full strength of NFLVR is not necessarily needed to solve
fundamental problems of valuation, hedging and portfolio
optimization. In addition, the notion of NFLVR is not robust with
respect to changes in the numeraire or the reference filtration and
it is not easy to check it in real markets. In parallel, there is
also the so-called {\sl Benchmark approach to Quantitative Finance}
(see e.g. \cite{PH}) that aims at developing a theory of valuation
that does not rely on the existence of an ELMM. For the hedging
problem in market models that do not admit an ELMM we refer to
\cite{Ruf}.

Weaker forms of no-arbitrage were thus introduced recently and it
turns out that the equivalent notions of no-arbitrage given by the
{\sl No-Arbitrage of the First Kind} (NA1) (see \cite{K12}, see also
\cite{LW} for an earlier version of this notion) and {\sl No
Unbounded Profit with Bounded Risk} (NUPBR) (see \cite{KK})
represent the minimal condition that still allows one to
meaningfully solve problems of pricing, hedging and portfolio
optimization. It is thus of interest to develop economically
significant market models, which satisfy NA1 (NUPBR), but not NFLVR.
As long as one remains within continuous market models, there may
not exist many models that satisfy NA1 (NUPBR) but not NFLVR. The
classical examples are based on Bessel processes (see e.g.
\cite{DS95}, \cite{PH}, \cite{H}). As hinted at in \cite{CL},
discontinuous market models may offer many more possibilities. To
this effect one may however point out that in \cite{K09} the author
shows that for exponential L\'evy models the various no-arbitrage
notions weaker than NFLVR are all equivalent to NFLVR. This still
leaves open the possibility of investigating the case of
jump-diffusion models, which is the setting that we shall consider
in this paper.

Working under NA1 (NUPBR) we cannot rely upon the existence of an
ELMM, nor upon the corresponding density process. In fact, if there
exists some form of arbitrage beyond NFLVR, then a possible
candidate density of an E$\sg$MM (ELMM) turns out to be a strict
local martingale. The density process can however be generalized to
the notion of an {\sl Equivalent Supermartingale Deflator} (ESMD)
which, considering a finite horizon $[0,T]$, is a process $D_t$ with
$D_0=1,\,D_t\ge 0,\,D_T>0\>P-a.s.$ and such that $D\bar V$ is a
supermartingale for all discounted value processes $\bar V_t$ of a
self-financing admissible portfolio strategy ($D_t$ is thus itself a
supermartingale). Notice that, if there exists an {\sl Equivalent
Supermartingale Measure} (ESMM), namely a measure $Q\sim P$ under
which all discounted self-financing portfolio processes are
supermartingales, then the process
$D_t:=\left(\frac{dQ}{dP}\right)_{\mid\F_t}$ is an ESMD that is
actually a martingale. An ESMD is however not necessarily a density
process, it may not even be a (local) martingale. Instead of
approaching directly the problem of constructing jump-diffusion
market models that satisfy NA1 (NUPBR) but not NFLVR, in this paper
we do it indirectly: for specific jump-diffusion market models that
satisfy NA1 (NUPBR) we construct an ESMD that is not a martingale,
not even a local martingale. We shall then show that for these
models the so constructed ESMD is the only candidate for the density
process of an ELMM (ESMM). If, thus, this only candidate is a
supermartingale that is not a martingale, then there cannot exist an
ELMM (ESMM). Since, furthermore, for these cases we shall show that
the physical measure $P$ cannot be an ELMM, the property of NFLVR
fails to hold.

A basic tool to obtain an ESMD is via the {\sl Growth Optimal
Portfolio} (GOP), which is a portfolio that outperforms any other
self-financing portfolio in the sense that the ratio between the two
processes is a supermartingale. For continuous markets it can in
fact be shown that, under local square integrability of the market
price of risk process, if $\bar V_t^*$ denotes the discounted value
process of the GOP, then $\hat Z_t:=\frac{1}{\bar V_t^*}$ is an ESMD
(see e.g. \cite{FR}). Furthermore, always in continuos markets and
under local square integrability of the market price of risk, the
undiscounted GOP process $V_t^*$ is a {\sl Numeraire Portfolio} in
the sense that self-financing portfolio values, expressed in units
of $V_t^*$, are supermartingales under the physical measure (in the
continuous case they are actually positive local martingales so that
NFLVR holds in the $V_t^*-$discounted market and $P$ itself is an
ELMM). This is however not true in general: the inverse of the
discounted GOP may fail to be even a local martingale (see e.g.
Example 5.1bis in \cite{KS99}). In particular, this may happen when
jumps are present (see e.g. \cite{B} Example 6, and \cite{CL}). On
the other hand, also for general semimartingale models one can show
(see \cite{KK}) the equivalence of \bd\item [i)] Existence of a
numeraire portfolio.
\item[ii)] Existence of an ESMD.
\item[iii)] Validity of NA1 (NUPBR).\ed
Since, when the GOP exists, this GOP is a numeraire portfolio even
if the inverse of its discounted value is not a local martingale
(see e.g. \cite{HS}), by the previous equivalence there exists then
an ESMD.  In this case, namely when the inverse of the discounted
GOP is not a local martingale, we then also have that the physical
measure is not an ELMM when using the GOP as numeraire. In general
this does not exclude that NFLVR may nevertheless hold (see
\cite{Ch}, see also \cite{T}). However, in the jump-diffusion case
we shall show that the process $\hat Z_t$ is the only candidate for
the density process of an ELMM (ESMM). Therefore, in this case the
failure of the inverse of the discounted GOP to be a (local)
martingale excludes the possibility of NFLVR.

In our quest for examples in the jump-diffusion setting, where the
GOP exists and thus NA1 (NUPBR) holds, but the inverse of the
discounted GOP may fail to be a local martingale, we shall start by
studying the existence and the properties of the GOP, thereby
focusing on the characteristics of the market price of risk vector
and showing that the existence of an ELMM (ESMM) depends strictly on
the relationship between the components of this vector and the
corresponding jump intensities. We then prove that the inverse of
the discounted GOP defines the only candidate to be the density
process of an ELMM and provide examples in which the discounted GOP
not only fails to be the inverse of a martingale density process,
but is a supermartingale that is not a local martingale and this
with and without constraints on the portfolio strategies.

The outline of the paper is as follows: In section \ref{S.1} we
describe our jump-diffusion market model and the admissible
strategies. Section \ref{S.2} is devoted to the existence and the
properties of the GOP as well as its relation with ESMDs. Examples
where the discounted GOP fails to be the inverse of a
(super)martingale density process are then discussed in section
\ref{S.3}. The model described in section \ref{S.1} as well as the
contents of subsection \ref{S.2.1} are based on the first part of
Chapter 14 in \cite{PH} (see also \cite{CP}). We recall them here to
make the presentation self-contained thereby giving also some
additional detail that will be needed later in the paper.

\section{The Jump Diffusion Market Model}\label{S.1}

Let there be given a complete probability space
$(\Omega,\mathcal{F},\textit{P})$, with a filtration
$\mathcal{F}=(\mathcal{F}_t)_{t\geq 0}$ satisfying the usual
conditions of right-continuity and completeness. We consider a
market containing $d\in\mathbb{N}$ sources of uncertainty.
Continuous uncertainty is represented by an \textit{m}-dimensional
standard Wiener process
$W=\{W_t=(W^1_t,\dots,W^m_t)^\top,\;t\in[0,\infty)\}$. Event driven
uncertainty on the other hand is modeled by an
(\textit{d-m})-variate point process, identified by the
$\mathcal{F}$-adapted counting process
$N=\{N_t=(N^1_t,\dots,N^{d-m}_t)^\top, \;t\in[0,\infty)\}$, whose
intensity
$\lambda=\{\lambda_t=(\lambda^1_t,\dots,\lambda^{d-m}_t)^\top,\;t\in[0,\infty)\}$
is a given, predictable and strictly positive process satisfying
$\lambda^k_t>0\quad\mbox{and}\quad\int_0^t\lambda^k_s ds<\infty$
almost surely for all $t\in[0,\infty)$ and $k\in\{1,2,\dots,d-m\}$.
We shall denote by $T_n$ the jump times of $N_t$. For each
univariate point process we can define a corresponding jump
martingale $M^k=\{M^k_t,\;t\in[0,\infty)\}$ which, following
\cite{PH}, we define via its stochastic differential \beq
\label{eq:M} dM_t^k=\frac{dN_t^k-\lambda_t^kdt}{\sqrt{\lambda_t^k}},
\eeq for all $t\in[0,\infty)$ and $k\in\{1,2,\dots,d-m\}$. We assume
that $W$ and $N$ are independent, generate all the uncertainty in
the model and that pairwise the $N^k$s do not jump at the same time.

The financial market consists of $d+1$ securities $S^j$, for
$j=0,1,\dots,d$, that model the evolution of wealth due to the
ownership of primary securities, with all income and dividends
reinvested. As usual, the first account $S^0_t$ is assumed to be
locally risk free, which means that it is of finite variation and
the solution of the differential equation
\[
dS^0_t=S^0_t r_t dt
\]
for $t\in[0,\infty)$, with $S^0_0=1$. The remaining assets $S^j$,
for $j=1,\dots,d$, are supposed to be risky and to be the solution
to the jump diffusion SDE \beq \label{eq:assets}
dS^j_t=S^j_{t-}\left(a^j_t dt+\sum_{k=1}^m b_t^{j,k}
dW_t^k+\sum_{k=m+1}^{d} b_t^{j,k} dM_t^{k-m}\right) \qquad\quad
S_0^j>0 \eeq for $t\in[0,\infty)$, where the short rate process $r$,
the appreciation rate processes $a^j$, the generalized volatility
processes $b^{j,k}$ and the intensity processes $\lambda^k$ are
almost surely finite and predictable. Assuming that these processes
are such that a unique strong solution to the system of
SDEs~\eqref{eq:assets} exists, we obtain the following explicit
expression for every $j=1,\dots,d$
\[
\begin{split}
S_t^j&= S_0^j\exp\left\{\int_0^t\left(a_s^j-\frac{1}{2}\sum_{k=1}^m(b_s^{j,k})^2\right)ds
+\sum_{k=1}^m\int_0^tb_s^{j,k}dW_s^k\right\}\cdot\\
&\cdot\prod_{k=m+1}^{d}\left[\exp\left\{-\int_0^tb_s^{j,k}\sqrt{\lambda_s^{k-m}}ds\right\}\prod_{n=1}^{N_t^{k-m}}
\left(\frac{b_{T_n}^{j,k}}{\sqrt{\lambda_{T_n}^{k-m}}}+1\right)\right]
\end{split}
\]
To ensure non-negativity for each primary security account, and
exclude jumps that would lead to negative values for $S_t^j$, we need
to make the following assumption:
\begin{ass} The condition
\[
b_t^{j,k}\geq-\sqrt{\lambda_t^{k-m}}
\]
holds for all $t\in[0,\infty)$, $j\in\{1,2,\dots,d\}$ and $k\in\{m+1,\dots,d\}$.
\label{ass:1}
\end{ass}
In what follows $b_t$ will denote the \textit{generalized
volatility} matrix $[b_t^{j,k}]_{j,k=1}^d$ for all $t\in[0,\infty)$
and we make the further assumption:
\begin{ass} The generalized volatility matrix $b_t$ is invertible for Lebesgue-almost-every $t\in[0,\infty)$.
\label{ass:2}
\end{ass}
The condition stated in Assumption~\ref{ass:2} means that no primary
security account can be formed as a portfolio of other primary
security accounts, i.e. the market does not contain redundant
assets.

We can now introduce the \textit{market price of risk} vector \beq
\label{eq:theta}
\theta_t=(\theta_t^1,\dots,\theta_t^d)^\top=b_t^{-1}[a_t-r_t\textbf{1}],
\eeq for $t\in[0,\infty)$. Here $a_t=(a_t^1,\dots,a_t^d)^\top$ is
the \textit{appreciation rate} vector and
$\textbf{1}=(1,\dots,1)^\top$ is the \textit{unit} vector.
Using~\eqref{eq:theta}, we obtain $a_t=b_t\theta_t+r_t\textbf{1}$ so
that we can rewrite the SDE~\eqref{eq:assets} in the form \beq
\label{eq:rew} dS^j_t=S^j_{t-}\left(r_t dt+\sum_{k=1}^m
b_t^{j,k}(\theta_t^kdt+ dW_t^k)+\sum_{k=m+1}^{d}
b_t^{j,k}(\theta_t^kdt+ dM_t^{k-m})\right), \eeq for
$t\in[0,\infty)$ and $j\in\{1,\dots,d\}$. For $k\in\{1,2,\dots,m\}$,
the quantity $\theta_t^k$ denotes the market price of risk with
respect to the $k$-th Wiener process $W^k$. In a similar way, if
$k\in\{m+1,\dots,d\}$, then $\theta_t^k$ can be interpreted as the
market price of the $(k-m)$-th event risk with respect to the
counting process $N^{k-m}$.

The vector process
$S=\{S_t=(S_t^0,\dots,S_t^d)^\top,\;t\in[0,\infty)\}$ characterizes
the evolution of all primary security accounts. In order to
rigorously describe the activity of trading in the financial market
we now recall the concept of \textit{trading strategy}. We emphasize
that we only consider self-financing trading strategies which
generate \textit{positive portfolio processes}.
\begin{df}  ({\sl Admissible strategies})$\qquad$
\begin{description}
\item[(a)] An admissible trading strategy with initial value $s$ is a $\mathbb{R}^{d+1}$-valued
predictable stochastic process
$\delta=\{\delta_t=(\delta_t^0,\dots,\delta_t^d)^\top,\;t\in[0,\infty)\}$,
where $\delta^j_t$ denotes the number of units of the $j$-th primary
security account held at time $t\in[0,\infty)$ in the portfolio, and
it is such that the It\^{o} integral $\int_0^T\delta^j_tdS^j_t$ is
well-defined for any $T>0$ and $j\in\{0,1,\dots,d\}$. Furthermore,
the value of the corresponding portfolio process at time $t$, which
as in \cite{PH} we denote by
$S_t^{s,\delta}=\sum_{j=0}^d\delta_t^jS_t^j$ with
$S_0^{s,\delta}=s>0$, is nonnegative for any $t\in[0,\infty)$. We
let $S_t^{\delta}:=S_t^{1,\delta}$ so that
$S_t^{s,\delta}=s\,S_t^{\delta}$. Since the strategy is
self-financing, we have also
$dS_t^{s,\delta}=\sum_{j=0}^d\delta_t^jdS_t^j. $
\item[(b)] For any admissible trading strategy $\delta$ the value of the corresponding discounted portfolio process
at time $t$ is defined as $\bar{S}_t^{s,\delta}=\frac{S_t^{s,\delta}}{S_t^0}$.
\end{description}
\end{df}

For a given strategy $\delta$ with strictly positive portfolio
process $S^{s,\delta}$ we denote as usual by $\pi_{\delta,t}^j$ the
fraction of wealth invested in the $j$-th primary security account
at time $t$, that is
$\pi_{\delta,t}^j=\delta_t^j\frac{S_{t-}^j}{S_{t-}^{s,\delta}}$ for
$t\in[0,\infty)$ and $j\in\{0,1,\dots,d\}$. In terms of the vector
of fractions
$\pi_{\delta,t}=(\pi_{\delta,t}^0,\dots,\pi_{\delta,t}^d)^\top$ we
obtain from~\eqref{eq:rew}, the self-financing property, and taking~\eqref{eq:M} into account, the following SDE for $S_t^{s,\delta}$
\beq \label{eq:conti}
\begin{split}
dS_t^{s,\delta}=S_{t-}^{s,\delta}\Biggl\{&\left(r_t+\sum_{k=1}^m\sum_{j=1}^d\pi_{\delta,t-}^jb_t^{j,k}\theta_t^k
+\sum_{k=m+1}^{d}\sum_{j=1}^d\pi_{\delta,t-}^jb_t^{j,k}\left(\theta_t^k-\sqrt{\lambda_t^{k-m}}\right)\right)dt\\
&+\sum_{k=1}^m\sum_{j=1}^d\pi_{\delta,t}^jb_t^{j,k}dW_t^k+\sum_{k=m+1}^{d}\sum_{j=1}^d\pi_{\delta,t}^j\frac{b_t^{j,k}}
{{\sqrt{\lambda_t^{k-m}}}}dN_t^{k-m}\Biggl\}
\end{split}
\eeq
Therefore, the value at time $t$ of the portfolio $S^{s,\delta}$ is
\beq
\label{eq:port}
\begin{split}
S_t^{s,\delta}&=s\exp\Biggl\{\int_0^t\left(r_s+\sum_{k=1}^d\sum_{j=1}^d\pi_{\delta,s}^jb_s^{j,k}
\theta_s^k-\frac{1}{2}\sum_{k=1}^m\left(\sum_{j=1}^d\pi_{\delta,s}^jb_s^{j,k}\right)^2\right)ds\\
&\;\;+\sum_{k=1}^m\int_0^t\sum_{j=1}^d\left(\pi_{\delta,s}^jb_s^{j,k}\right)dW_s^k\Biggl\}\cdot
\prod_{k=m+1}^d\Biggl[\exp\left\{-\int_0^t\sum_{j=1}^d\pi_{\delta,s}^jb_s^{j,k}\sqrt{\lambda_s^{k-m}}ds\right\}\\
&\;\;\cdot\prod_{n=1}^{N_t^{k-m}}\left(\frac{\sum_{j=1}^d\pi_{\delta,T_n}^jb_{T_n}^{j,k}}{\sqrt{\lambda_{T_n}^{k-m}}}+1\right)\Biggl]
\end{split}
\eeq From the last term on the right hand side in~\eqref{eq:port} it
is immediately seen that a portfolio process remains strictly
positive if and only if \beq \label{eq:ammiss} \sum_{j=1}^d
\pi_{\delta,t}^jb_t^{j,k}>-\sqrt{\lambda_t^{k-m}}\quad\mbox{a.s.}
\eeq for all $k\in\{m+1,\dots,d\}$ and $t\in[0,\infty)$. This
condition is guaranteed by Assumption \ref{ass:1}.

\section{The Growth Optimal Portfolio (GOP)}\label{S.2}

\subsection{Derivation of the GOP and its dynamics}\label{S.2.1}

We start from
\begin{df} For an admissible trading strategy $\delta$, leading to a strictly positive portfolio process,
the {\bf growth rate process} $g^\delta=(g^\delta_t)_{t\geq0}$ is
defined as the drift term in the SDE satisfied by the process
$\log{S^\delta}=(\log{S^\delta_t})_{t\geq0}$. An admissible trading
strategy $\delta^*$ (and the corresponding portfolio process
$S^{\delta^*}$) is said to be growth-optimal if $g_t^{\delta^*}\geq
g_t^\delta$ P-a.s. for all $t\in[0,\infty)$ for any admissible
trading strategy $\delta$. We shall use the acronym $GOP$ to denote
the growth-optimal portfolio.\end{df}

By applying It\^o's formula and suitably adding and subtracting
terms (see \cite{M}) one immediately obtains the general expression
for $g_t^\delta$, namely
\begin{lm} For any admissible trading strategy $\delta$, the SDE satisfied by $\log S^{\delta}$ is
\beq \label{eq:g} d\log{S_t^\delta}=g_t^\delta
dt+\sum_{k=1}^m\sum_{j=1}^d\pi_{\delta,t}^jb_t^{j,k}dW_t^k
+\sum_{k=m+1}^d\log{\left(1+\sum_{j=1}^d\pi_{\delta,t}^j\frac{b_t^{j,k}}{\sqrt{\lambda_t^{k-m}}}\right)}\sqrt{\lambda_t^{k-m}}dM_t^{k-m}
\eeq where $g_t^\delta$ is the growth rate given by \beq
\label{eq:gg}
\begin{split}
g_t^\delta &=r_t+\sum_{k=1}^m\left[\sum_{j=1}^d\pi_{\delta,t}^j b_t^{j,k}\theta_t^k-
\frac{1}{2}\left(\sum_{j=1}^d\pi_{\delta,t}^j b_t^{j,k}\right)^2\right] \\
&+\sum_{k=m+1}^d\left[\sum_{j=1}^d\pi_{\delta,t}^j
b_t^{j,k}\left(\theta_t^k-\sqrt{\lambda_t^{k-m}}\right)
+\log{\left(1+\sum_{j=1}^d\pi_{\delta,t}^j\frac{b_t^{j,k}}{\sqrt{\lambda_t^{k-m}}}\right)}\lambda_t^{k-m}\right]
\end{split}
\eeq
for $t\in[0,\infty)$.
\label{lm:lma}
\end{lm}

In order to obtain the GOP dynamics we now maximize separately the
two sums on the right hand side of~\eqref{eq:gg} with respect to the
\textit{portfolio volatilities} $c_t^k:=\sum_{j=1}^d\pi_{\delta,t}^j
b_t^{j,k}$ for $k\in\{1,\dots,d\}$. Note that for the first sum a
unique maximum exists, because it is a negative definite quadratic
form with respect to the portfolio volatilities. In order to
guarantee the existence and the uniqueness of a maximum also in the
second sum, we have to impose the following condition
\begin{ass} The intensities and the market price of event risk components satisfy
\[
\sqrt{\lambda_t^{k-m}}>\theta_t^k
\]
for all $t\in[0,\infty)$ and $k\in\{m+1,\dots,d\}$.
\label{ass:test}
\end{ass}
This is because the first derivative of \beq \label{eq:varg}
c^k_t\left(\theta_t^k-\sqrt{\lambda_t^{k-m}}\right)+\log{\left(1+\frac{c^k_t}{\sqrt{\lambda_t^{k-m}}}\right)}\lambda_t^{k-m}
\eeq with respect to $c_t^k$, which is
$\left(\theta_t^k-\sqrt{\lambda_t^{k-m}}\right)+\frac{\lambda_t^{k-m}}{\sqrt{\lambda_t^{k-m}}+c^k_t}$,
is positive for all $c^k_t>-\sqrt{\lambda_t^{k-m}}$ if
Assumption~\ref{ass:test} does not hold. These, by virtue
of~\eqref{eq:ammiss}, are precisely all the possible values of the
\textit{event driven volatility} $\sum_{j=1}^d\pi_{\delta,t}^j
b_t^{j,k}$, $k\in\{m+1,\dots,d\}$. Therefore if
Assumption~\ref{ass:test} fails to hold there will not exist an
optimal growth rate, since~\eqref{eq:varg} tends to infinity as
$c_t^k\to\infty$, for any $k\in\{m+1,\dots,d\}$.

This condition allows us to introduce the predictable vector process
$c_t^*=(c_t^{*1},\dots,c_t^{*d})^\top$ which describes the optimal
generalized portfolio volatilities. For the components with
$k\in\{1,\dots,m\}$, we get from the first order condition that
identifies the maximum growth rate the following
\[
\theta_t^k-c_t^k=0\Longleftrightarrow c_t^k=\theta_t^k
\]
For the last ($d-m$) components we get, again from the first order
condition, noticing that the function of $c_t^k$ in~\eqref{eq:varg}
is strictly concave and that we must have
$c_t^k>-\sqrt{\lambda_t^{k-m}}$,
\[
\left(\theta_t^k-\sqrt{\lambda_t^{k-m}}\right)+\frac{\lambda_t^{k-m}}{\sqrt{\lambda_t^{k-m}}+c^k_t}=0
\Longleftrightarrow
c_t^k=\frac{\theta_t^k}{1-\theta_t^k(\lambda_t^{k-m})^{-\frac{1}{2}}}
\]
Therefore the vector $c_t^*$ has the following representation
\beq
\label{eq:c}
c_t^{*k}=
\begin{cases}
\;\;\;\;\;\;\;\;\;\theta_t^k & \qquad\text{for} \quad k\in\{1,2,\dots,m\}\\
\frac{\theta_t^k}{1-\theta_t^k(\lambda_t^{k-m})^{-\frac{1}{2}}}
&\qquad \text{for} \quad k\in\{m+1,\dots,d\}
\end{cases}
\eeq for $t\in[0,\infty)$. Note that a very large jump intensity
with $\lambda_t^{k-m}\gg1$ or
$\frac{\theta_t^k}{\sqrt{\lambda_t^{k-m}}}\ll1$ causes the
corresponding component $c_t^{*k}$ to approach the market price of
jump risk $\theta_t^k$ asymptotically for given $t\in[0,\infty)$ and
$k\in\{m+1,\dots,d\}$. In this case the structure of the components
$c_t^{*k}\approx\theta_t^k$ for $k\in\{m+1,\dots,d\}$ is similar to
those obtained with respect to the Wiener processes. Intuitively
this is because, when jumps occur more and more frequently, almost
continuously, the jump martingales $M^k$ become nearly
indistinguishable from the continuous ones.

The above considerations lead immediately to the following (see also
Corollary 14.1.5 in \cite{PH})
\begin{lm} Under the Assumptions~\ref{ass:1},~\ref{ass:2} and~\ref{ass:test} the fractions
\beq
\label{eq:fract}
\pi_{\delta_*,t}=(\pi_{\delta_*,t}^1,\dots,\pi_{\delta_*,t}^d)=(c_t^{*\top}b_t^{-1})^\top
\eeq
determine uniquely the GOP and the corresponding portfolio process $S^{\delta_*}=\{S^{\delta_*}_t,\;t\in[0,\infty)\}$ satisfies the SDE
\beq
\label{eq:gop}
\begin{split}
dS^{\delta_*}_t= S^{\delta_*}_{t-}\Biggl(&r_t dt +\sum_{k=1}^m\theta_t^k(\theta_t^kdt+dW_t^k)\\
+&\sum_{k=m+1}^d\frac{\theta_t^k}{1-\theta_t^k(\lambda_t^{k-m})^{-\frac{1}{2}}}(\theta_t^kdt+dM_t^{k-m})\Biggl),
\end{split}
\eeq for $t\in[0,\infty)$, with $S_0^{\delta_*}>0$. Note that
Assumption~\ref{ass:test} guarantees that the portfolio process
$S^{\delta_*}$ is strictly positive.\end{lm}

By~\eqref{eq:gg},~\eqref{eq:c} and~\eqref{eq:fract} we obtain the
optimal growth rate of the GOP in the form
\[
\begin{split}
g_t^{\delta_*} &=r_t+\sum_{k=1}^m\left[\left(\theta_t^k\right)^2-\frac{1}{2}\left(\theta_t^k\right)^2\right]
+\sum_{k=m+1}^d\Biggl[\frac{\theta_t^k}{1-\theta_t^k(\lambda_t^{k-m})^{-\frac{1}{2}}}\left(\theta_t^k-\sqrt{\lambda_t^{k-m}}\right)\\
&\;+\log{\left(1+\frac{\theta_t^k}{1-\theta_t^k(\lambda_t^{k-m})^{-\frac{1}{2}}}\frac{1}{\sqrt{\lambda_t^{k-m}}}\right)}\lambda_t^{k-m}\Biggl]\\
&=r_t+\frac{1}{2}\sum_{k=1}^m\left(\theta_t^k\right)^2+\sum_{k=m+1}^d\lambda_t^{k-m}\left(\log{\left(1+\frac{\theta_t^k}
{\sqrt{\lambda_t^{k-m}}-\theta_t^k}\right)}-\frac{\theta_t^k}{\sqrt{\lambda_t^{k-m}}}\right)
\end{split}
\]
for $t\in[0,\infty)$.

\subsection{GOP and Martingale Deflators}\label{S.2.2}
In what follows we consider a fixed finite time horizon $T<\infty$
and investigate whether the model introduced above represents a
viable financial market, in particular we shall check whether
properly defined arbitrage opportunities are excluded. First we
recall the following
\begin{df}\label{ESMM}
An {\sl equivalent local martingale measure} (ELMM) is a measure
$Q\sim P$ such that all price processes, expressed in units of the
risk-free asset (i.e. discounted prices), are local martingales.
Analogously, $Q\sim P$ is an {\sl equivalent supermartingale
measure} (ESMM) if the price processes, expressed in units of the
risk-free asset, are supermartingales.\end{df} From \cite{K12} we
also recall the
\begin{df} An $\mathcal{F}_T$-measurable nonnegative random variable $\xi$ is called arbitrage of
the first kind if $P(\xi>0)>0$ and, for all initial values $s\in(0,\infty)$, there exists an
admissible trading strategy $\delta$ such that $\bar{S}_T^{s,\delta}\geq\xi$ $P$-a.s.
We say that the financial market is viable if there are no arbitrages of the first kind, i.e. the condition NA1 holds.
\end{df}
We next show that a sufficient condition for the absence of
arbitrages of the first kind is the existence of a
\textit{supermartingale deflator} which we describe as
\begin{df} An equivalent supermartingale deflator (ESMD) is a real-valued nonnegative adapted process $D=(D_t)_{0\leq t\leq T}$
with $D_0=1$ and $D_T>0$ $P$-a.s. and such that the process
$D\bar{S}^\delta=(D_t\bar{S}^\delta_t)_{0\leq t\leq T}$ is a
supermartingale for every admissible trading strategy $\delta$. We
denote by $\mathcal{D}$ the set of all supermartingale deflators.
(By taking $\delta\equiv (1,0,\cdots,0)$, one has that $D$ is itself
a supermartingale).\label{df:deflator}
\end{df}

\begin{prop} If $\mathcal{D}\neq\emptyset$ then there cannot exist arbitrages of the first kind.
\end{prop}
\begin{proof}
(adapted from~\cite{FR}). Let $D\in\mathcal{D}$ and suppose that
there exists a random variable $\xi$ yielding an arbitrage of the
first kind. Then, for every $n\in\mathbb{N}$, there exists an
admissible trading strategy $\delta^n$ such that,
$\bar{S}_T^{1/n,\delta^n}\geq\xi$ $P$-a.s. For every
$n\in\mathbb{N}$, the process
$D\bar{S}^{1/n,\delta^n}=(D_t\bar{S}^{1/n,\delta^n}_t)_{0\leq t\leq
T}$ is a supermartingale. So, for every $n\in\mathbb{N}$ one has
$E[D_T\xi]\leq E[D_T\bar{S}^{1/n,\delta^n}_T]\leq
E[D_0\bar{S}^{1/n,\delta^n}_0]=\frac{1}{n}.$ Letting $n\to\infty$
gives $E[D_T\xi]=0$ and hence $D_T\xi=0$ $P$-a.s. Since, due to
Definition~\ref{df:deflator}, we have $D_T>0$ $P$-a.s. this implies
that $\xi=0$ $P$-a.s., which contradicts the assumption that $\xi$
is an arbitrage of the first kind.
\end{proof}
In the remaining part of this section we will derive a fundamental
property of the GOP. Dealing with GOP-denominated portfolio
processes, following \cite{PH} we first introduce the following
notation.
\begin{df}\label{benchm}  For any portfolio process $S^\delta$, the process
$\hat{S}^\delta=\left(\hat{S}_t^\delta\right)_{0\leq t\leq T}$, defined as
$\hat{S}_t^\delta:=S^\delta_t/S_t^{\delta_*}$ for $t\in[0,T]$, is called {\sl Benchmarked portfolio process}.
\end{df}
\begin{rem}
We shall often refer to the inverse of the discounted GOP as
$\hat{Z}_t:=\frac{1}{\bar{S}_t^{\delta_*}}.$
\end{rem}
We begin with the following proposition.
\begin{prop} Under Assumptions~\ref{ass:1},~\ref{ass:2} and~\ref{ass:test} the discounted GOP process
$\bar{S}^{\delta_*}=\{\bar{S}^{\delta_*}_t,\;t\in[0,T]\}$ is the inverse of a supermartingale deflator $D$.
Equivalently, $\hat{Z}_t$ is a supermartingale deflator.
\label{proposition:num}
\end{prop}
\begin{proof} It suffices to show that every benchmarked portfolio is a local
martingale that, being nonnegative, is then a supermartingale by
Fatou's lemma.  According to the product formula (see Corollary II-2
in~\cite{P}) we first have that
\[
d\left(\frac{S_t^\delta}{S_t^{\delta_*}}\right)=d\left(\frac{\bar{S}_t^\delta}{\bar{S}_t^{\delta_*}}\right)
=\frac{d\bar{S}_t^\delta}{\bar{S}_{t-}^{\delta_*}}+d\left(\frac{1}{\bar{S}_{t-}^{\delta_*}}\right)\bar{S}_t^\delta
+d\left[ \bar{S}_t^\delta,\frac{1}{\bar{S}_t^{\delta_*}}\right]
\]
For the term involving  $d\bar{S}_t^\delta$  note that, from
(\ref{eq:conti}) and taking into account (\ref{eq:M}) as well as the
definition of the investment ratios $\pi_{\delta,t}$, one obtains
\beq \label{eq:S}
\begin{split}
d\bar{S}_t^\delta&=\bar{S}_{t-}^\delta\left\{\sum_{k=1}^m\left(\sum_{j=1}^d\delta_t^j\frac{S_t^j}
{S_t^\delta}b_t^{j,k}\right)\left(\theta_t^kdt+dW_t^k\right)+\sum_{k=m+1}^{d}
\left(\sum_{j=1}^d\delta_t^j\frac{S_{t-}^j}{S_{t-}^\delta}b_t^{j,k}\right)\left(\theta_t^kdt+dM_t^{k-m}\right)\right\}\\
&=\sum_{k=1}^m\left(\sum_{j=1}^d\delta_t^j\bar{S}_t^jb_t^{j,k}\right)\left(\theta_t^kdt+dW_t^k\right)
+\sum_{k=m+1}^{d}\left(\sum_{j=1}^d\delta_t^j\bar{S}_{t-}^jb_t^{j,k}\right)\left(\theta_t^kdt+dM_t^{k-m}\right)
\end{split}
\eeq Next, using It\^o's formula, from (\ref{eq:gop}) we obtain for
the inverse of the discounted GOP
\[
\begin{split}
d\left(\frac{1}{\bar{S}_t^{\delta_*}}\right)=\Biggl[&-\frac{1}{\bar{S}_t^{\delta_*}}\left(\sum_{k=1}^m
\left(\theta_t^k\right)^2+\sum_{k=m+1}^d\frac{\theta_t^k}{1-\frac{\theta_t^k}{\sqrt{\lambda_t^{k-m}}}}
\left(\theta_t^k-\sqrt{\lambda_t^{k-m}}\right)\right)\\
&+\frac{1}{\bar{S}_t^{\delta_*}}\sum_{k=1}^m\left(\theta_t^k\right)^2\Biggl]dt-\frac{1}{\bar{S}_t^{\delta_*}}\sum_{k=1}^m\theta_t^kdW_t^k\\
&+\frac{1}{\bar{S}_{t-}^{\delta_*}}\sum_{k=m+1}^d\left[\left(1+\frac{\theta_t^k}{1-\frac{\theta_t^k}
{\sqrt{\lambda_t^{k-m}}}}\frac{1}{\sqrt{\lambda_t^{k-m}}}\right)^{-1}-1\right]dN_t^{k-m}
\end{split}
\]
The terms of the last sum can be simplified in the following way
\[
\frac{1}{S_{t-}^{\delta_*}}\left(\left(\frac{\sqrt{\lambda_t^{k-m}}}{\sqrt{\lambda_t^{k-m}}-\theta_t^k}\right)^{-1}-1\right)
=-\frac{\theta_t^k}{\bar{S}_{t-}^{\delta_*}\sqrt{\lambda_t^{k-m}}}
\]
so that, by adding and subtracting
$-\sum_{k=m+1}^d\frac{\theta_t^k\sqrt{\lambda_t^{k-m}}}{\bar{S}_t^{\delta_*}}dt$
and rearranging all the terms, we obtain \beq \label{eq:inverse}
d\left(\frac{1}{\bar{S}_t^{\delta_*}}\right)=-\frac{1}{\bar{S}_t^{\delta_*}}\sum_{k=1}^m\theta_t^kdW_t^k
-\frac{1}{\bar{S}_{t-}^{\delta_*}}\sum_{k=m+1}^d\theta_t^kdM_t^{k-m}.
\eeq Finally, the last term is just
\[
d\left[ \bar{S}_t^\delta,\frac{1}{\bar{S}_t^{\delta_*}}\right] =
-\sum_{k=1}^m
\left(\sum_{j=1}^d\delta_t^j\hat{S}_t^jb_t^{j,k}\right)\theta_t^kdt-\sum_{k=m+1}^d
\left(\sum_{j=1}^d\delta_t^j\hat{S}_{t-}^jb_t^{j,k}\right)\frac{\theta_t^k}{\lambda_t^{k-m}}dN_t^{k-m}
\]
Summing up all the components we obtain
\[
\begin{split}
d\left(\frac{\bar{S}_t^\delta}{\bar{S}_t^{\delta_*}}\right)&=
\sum_{k=1}^m\left(\sum_{j=1}^d\delta_t^j\hat{S}_t^jb_t^{j,k}-\hat{S}_t^\delta\theta_t^k\right)dW_t^k
+\sum_{k=m+1}^d\left(\sum_{j=1}^d\delta_t^j\hat{S}_{t-}^jb_t^{j,k}-\hat{S}_{t-}^\delta\theta_t^k\right)dM_t^{k-m}\\
&+\sum_{k=m+1}^d\left(\sum_{j=1}^d\delta_t^j\hat{S}_t^jb_t^{j,k}\right)\theta_t^kdt-\sum_{k=m+1}^d
\left(\sum_{j=1}^d\delta_t^j\hat{S}_{t-}^jb_t^{j,k}\right)\frac{\theta_t^k}{\lambda_t^{k-m}}dN_t^{k-m}
\end{split}
\]
from which, observing that
\[
\begin{split}
\sum_{k=m+1}^d&\left(\sum_{j=1}^d\delta_t^j\hat{S}_t^jb_t^{j,k}\right)\theta_t^kdt-\sum_{k=m+1}^d
\left(\sum_{j=1}^d\delta_t^j\hat{S}_{t-}^jb_t^{j,k}\right)\frac{\theta_t^k}{\lambda_t^{k-m}}dN_t^{k-m}=\\
=&-\sum_{k=m+1}^d\left(\sum_{j=1}^d\delta_t^j\hat{S}_{t-}^jb_t^{j,k}\right)\frac{\theta_t^k}
{\sqrt{\lambda_t^{k-m}}}\left(\frac{dN_t^{k-m}-\lambda_t^{k-m}dt}{\sqrt{\lambda_t^{k-m}}}\right)
\end{split}
\]
we obtain
\[
\begin{split}
d\left(\frac{\bar{S}_t^\delta}{\bar{S}_t^{\delta_*}}\right)=&\sum_{k=1}^m\left(\sum_{j=1}^d
\delta_t^j\hat{S}_t^jb_t^{j,k}-\hat{S}_t^\delta\theta_t^k\right)dW_t^k\\
+&\sum_{k=m+1}^d\left(\left(\sum_{j=1}^d\delta_t^j\hat{S}_{t-}^jb_t^{j,k}\right)\left(1-\frac{\theta_t^k}
{\sqrt{\lambda_t^{k-m}}}\right)-\hat{S}_{t-}^\delta\theta_t^k\right)dM_t^{k-m}
\end{split}
\]
which, indeed,  is a nonnegative local martingale and thus a
supermartingale. \end{proof}

From this proof it actually follows that $\hat{Z}_t$ is an
equivalent local martingale deflator (ELMD) in the sense that
$\hat{Z}_t\,\bar S^{\delta}$ are local martingales.

Generalizing the notion of Benchmarked portfolio process (see Definition \ref{benchm})  we recall the following
\begin{df} An admissible portfolio process $S^{\tilde{\delta}}=\left(S_t^{\tilde{\delta}}\right)_{0\leq t\leq T}$
has the numeraire property if all admissible portfolio processes $S^\delta=\left(S^\delta_t\right)_{0\leq t\leq T}$,
when denominated in terms of $S^{\tilde{\delta}}$, are supermartingales, i.e. if the process
$S^\delta/S^{\tilde{\delta}}=\left(S^\delta_t/S^{\tilde{\delta}}_t\right)_{0\leq t\leq T}$ is a supermartingale
for every admissible trading strategy $\delta$.
\label{df:numeraire}
\end{df}
\begin{rem}
As a corollary of Proposition~\ref{proposition:num} we have that the GOP has the numeraire property.
\label{rem:rem}
\end{rem}
In continuous financial markets it can be shown that the numeraire
portfolio is unique (see e.g. \cite{FR}). The proof in \cite{FR} can
be carried over rather straightforwardly to the jump-diffusion case
(see \cite{M}) so that we have
\begin{prop} The numeraire portfolio process $S^{\tilde{\delta}}=\left(S_t^{\tilde{\delta}}\right)_{0\leq t\leq T}$
is unique (in the sense of indistinguishability). Furthermore, there exists an unique admissible trading strategy
$\tilde{\delta}$ such that $S^{\tilde{\delta}}$ is the numeraire portfolio, up to a null subset of $\Omega\times[0,T]$.
\label{proposition:num2}
\end{prop}
\begin{rem}
Since, as we have seen, the GOP has the numeraire property and the numeraire portfolio is unique, the GOP is the unique numeraire portfolio.
\label{remark:rem2}
\end{rem}
The numeraire property of the GOP plays a crucial role in the
concept of \textit{real world pricing} allowing one, also in the
present jump-diffusion setting, to perform pricing of contingent
claims in financial markets for which no ELMM may exist
(see~\cite{PH}, \cite{FR}).

\section{Supermartingale Deflators/Densities}\label{S.3}

Due to Propositions~\ref{proposition:num} and~\ref{proposition:num2}
(see Remarks~\ref{rem:rem} and~\ref{remark:rem2}), the GOP coincides
with the unique numeraire portfolio and its discounted value is also
the inverse of the supermartingale deflator $\hat{Z}_t$. This means
that, if we express all price processes in terms of the GOP, the
original probability measure $P$ becomes an ESMM. We now investigate
particular cases in which $\hat{Z}_t$ is not a martingale and so,
since it will be shown to be also the only candidate for the density
process of an ELMM, for these cases NFLVR fails to hold.
Furthermore, when expressing the price processes in terms of the
GOP, the physical measure $P$ is not an ELMM.

We start by showing that for our financial market model the inverse
of the discounted GOP is the only possible local martingale
deflator; thus it is also the only candidate to be the Radon-Nikodym
derivative (density process) of an ELMM. These concepts are in fact
strictly linked to one another, since a supermartingale deflator $D$
defines an ELMM if and only if $D_T$ integrates to $1$. On the other
hand, the Radon-Nikodym derivative of an ELMM is of course a
supermartingale deflator. We prove the claim of the uniqueness by
studying changes of measure via the general Radon-Nikodym derivative
when dealing with a jump-diffusion process, namely (see \cite{Br},
\cite{R}) \beq \label{eq:radon}
\begin{split}
L_t=\exp&\left\{-\frac{1}{2}\sum_{k=1}^m\int_0^t\left(\varphi_s^k\right)^2ds+\sum_{k=1}^m\int_0^t\varphi_s^kdW_s^k\right\}\cdot\\
&\cdot\prod_{k=m+1}^d\left\{\exp\left[\int_0^t\left(1-\psi_s^{k-m}\right)\lambda_s^{k-m}ds\right]\prod_{n=1}^{N_t^{k-m}}\psi_{T_n}^{k-m}\right\}
\end{split}
\eeq where $\varphi_t$ is a square integrable predictable process
and $\psi_t$ is a positive predictable process, integrable with
respect to $\lambda_t$. We will show that these coefficients have to
satisfy a linear system whose only solution leads to the same
dynamics as for the inverse of the discounted GOP.
\begin{prop} Under Assumptions~\ref{ass:1},~\ref{ass:2} and ~\ref{ass:test}, the inverse of the discounted GOP, namely $\hat{Z}_t$,
is the only candidate to be the Radon-Nikodym derivative of an ELMM.
Equivalently, it is the only ELMD in the sense of what was specified
after the statement of Proposition \ref{proposition:num}.
\end{prop}
\begin{proof} We start by defining the Wiener and the Poisson martingales $W_t^Q$ and $M_t^Q$ under
the new measure $Q$ defined by $L_t$ in~\eqref{eq:radon}:
\[
\begin{cases}
\;\;dW_t^{Q,k}\;\;\;\;\;=dW_t^k\;\;\;-\varphi_t^kdt& \qquad\text{for} \quad k\in\{1,2,\dots,m\}\\
\;\;dM_t^{Q,k-m}=dN_t^{k-m}-\psi_t^{k-m}\lambda_t^{k-m}dt &\qquad
\text{for} \quad k\in\{m+1,\dots,d\}
\end{cases}
\]
thereby obtaining the SDEs satisfied by the primary security accounts
\beq
\label{eq:EL}
\begin{split}
dS_t^j=S_{t-}^j&\Biggl\{\left(r_t+\sum_{k=1}^db_t^{j,k}\theta_t^k+\sum_{k=1}^mb_t^{j,k}\varphi_t^k
+\sum_{k=m+1}^db_t^{j,k}\psi_t^{k-m}\sqrt{\lambda_t^{k-m}}-\sum_{k=m+1}^db_t^{j,k}\sqrt{\lambda_t^{k-m}}\right)dt\\
&+\sum_{k=1}^mb_t^{j,k}dW_t^{Q,k}+\sum_{k=m+1}^d\frac{1}{\sqrt{\lambda_t^{k-m}}}b_t^{j,k}dM_t^{Q,k-m}\Biggl\}
\end{split}
\eeq for every $j\in\{1,\dots,d\}$ and $t\in[0,T]$. If $L_t$ is the
Radon-Nikodym derivative of an ELMM, the drift term in~\eqref{eq:EL}
must be equal to $r_t$ for every $t\in[0,T]$, so that the
coefficients $\varphi_t$ and $\psi_t$ must be the solution to the
linear system defined by the following equations
\[
\sum_{k=1}^mb_t^{j,k}\varphi_t^k+\sum_{k=m+1}^db_t^{j,k}\psi_t^{k-m}\sqrt{\lambda_t^{k-m}}
=-\sum_{k=1}^mb_t^{j,k}\theta_t^k-\sum_{k=m+1}^db_t^{j,k}\theta_t^k+\sum_{k=m+1}^db_t^{j,k}\sqrt{\lambda_t^{k-m}}
\]
for every $j\in\{1,\dots,d\}$ and $t\in[0,T]$. Since, by virtue of
the standing Assumption~\ref{ass:2} the generalized volatility
matrix $b_t$ has full rank, the linear system admits an unique
solution which is given by
\[
\begin{cases}
\;\;\varphi_t^k\;\;\;\;\;=-\theta_t^k, & \qquad\text{for} \quad k\in\{1,2,\dots,m\}\\
\;\;\psi_t^{k-m}=1-\frac{\theta_t^k}{\sqrt{\lambda_t^{k-m}}},
&\qquad \text{for} \quad k\in\{m+1,\dots,d\}
\end{cases}
\]
Plugging these terms into~\eqref{eq:radon} we see that the Radon-Nikodym derivative has thus the following expression
\beq
\label{eq:LL}
\begin{split}
L_t=\exp&\left\{-\frac{1}{2}\sum_{k=1}^m\int_0^t\left(\theta_s^k\right)^2ds-\sum_{k=1}^m\int_0^t\theta_s^kdW_s^k\right\}\\
&\prod_{k=m+1}^d\left\{\exp\left[\int_0^t\theta_t^k\sqrt{\lambda_s^{k-m}}ds\right]
\prod_{n=1}^{N_t^{k-m}}\left(1-\frac{\theta_{T_n}^k}{\sqrt{\lambda_{T_n}^{k-m}}}\right)\right\}
\end{split}
\eeq which is precisely the inverse of the discounted GOP, since
(see the equation~\eqref{eq:inverse} satisfied by
$(\bar{S}_t^{\delta_*})^{-1}$)
\[
dL_t=-L_{t-}\left(\sum_{k=1}^m\theta_t^kdW_t^k+\sum_{k=m+1}^d\theta_t^kdM_t^{k-m}\right)
\]
\end{proof}
\begin{rem} Since the vector $\psi_t$ in~\eqref{eq:radon} has to be positive for every $t\in[0,T]$, we
note that the condition
$1-\frac{\theta_t^k}{\sqrt{\lambda_t^{k-m}}}\geq0$ has to hold true
for every $k\in\{m+1,\dots,d\}$. This means that if
Assumption~\ref{ass:test} is violated, and there is at least one
$k\in\{m+1,\dots,d\}$ for which $\sqrt{\lambda_t^{k-m}}<\theta_t^k$,
there cannot exist an ELMM.
\end{rem}
\begin{rem}Notice that, by analogy to the 2nd FTAP, for our complete
market here we have uniqueness of the ELMD.
\end{rem}
We want to emphasize that, thanks to the proposition just proved,
the financial market does not admit an ELMM as soon as the inverse
of the discounted GOP is a strict supermartingale that is not a
martingale. We therefore try to find some particular cases where the
process $\hat{Z}_t$ is a supermartingale that is not a martingale
(not even a local martingale). In the continuous case the most
notorious example of a martingale deflator that does not yield an
ELMM is the three-dimensional Bessel process $\beta_t$: the GOP is
simply obtained by placing all the wealth in the stock, but it has
infinite expected growth rate and $\frac{1}{\beta_t}$ fails to
integrate to $1$. In the jump diffusion market that we are
considering we note that, as $\theta_t^k\to\sqrt{\lambda_t^{k-m}}$
for any $k\in\{m+1,\dots,d\}$, the GOP would explode and therefore
$\hat{Z}_t$ would be close to zero as soon as a jump occurs
(see~\eqref{eq:LL}). In this case, $\hat{Z}_T$ will not integrate to
$1$, thus preventing the existence of an ELMM, and the GOP cannot be
used as a numeraire.

\subsection{Supermartingale Densities that are not (local) Martingale Densities}

Let us see what happens if we violate Assumption~\ref{ass:test},
namely if we let $\sqrt{\lambda_t^{k-m}}$ be less than $\theta_t^k$.
For simplicity, in this entire section we will consider the
particular case in which $d=2$ and $m=1$.

In order to be able to maximize the second sum in~\eqref{eq:gg} we
then have to impose a restriction on the possible trading strategies
in the form of the following assumption.
\begin{ass}
There exists a positive real number $\psi$ such that for any
admissible trading strategy
$\pi=\{\pi_t=\left(\pi_t^0,\pi_t^1,\pi_t^2\right)^\top,\;t\in[0,T]\}$
we have \beq \label{eq:ex} \pi_t^1 b_t^{1,2}+\pi_t^2
b_t^{2,2}\leq\psi, \eeq for all $t\in[0,T]$ \label{ass:ex}
\end{ass}
\begin{rem}
Assumption~\ref{ass:ex} can be seen as a \textit{convex constraint}
limiting the portfolio volatility belonging to the jump martingale
$M$. This condition has a clearer financial interpretation when
$b_t^2:=b_t^{1,2}=b_t^{2,2}$ for each $t\in[0,T]$. In this case we
get that~\eqref{eq:ex} is equivalent to a constraint on the minimum
amount of wealth invested in the risk free asset
$\pi_t^0\geq1-\frac{\psi}{b_t^2}.$ We emphasize that this applies
also to the case in which $b_t^2=-\sqrt{\lambda_t}$ for each
$t\in[0,T]$, in which jumps are used to describe the default of the
primary accounts.
\end{rem}
In this case the optimal generalized portfolio volatilities are described by the following predictable process
\beq
\label{eq:cmod}
\tilde{c}_t^k=
\begin{cases}
\;\;\theta_t^1 & \qquad\text{for} \quad k=1\\
\;\;\psi &\qquad \text{for} \quad k=2.
\end{cases}
\eeq The first component $\tilde{c}_k^1$ follows from the first
order conditions identifying the maximum growth rate, while
$\tilde{c}_k^2$ is the maximum value obtainable in the constrained
setting. From~\eqref{eq:conti} we have that, for the case when
$c_t^k$ are given by~\eqref{eq:cmod}, the discounted GOP must
satisfy the following SDE \beq \label{eq:gopex}
d\bar{S}_t^{\delta_*}=\bar{S}_{t-}^{\delta_*}\left\{\theta_t^1\left(\theta_t^1dt+dW_t\right)+\psi\left(\theta_t^2dt+dM_t\right)\right\}.
\eeq The convex constraints just introduced are the framework that
enables us to provide a simple example of a market in which GOP
denominated prices are strict supermartingales, as we show in the
next theorem.
\begin{thm} Under Assumptions~\ref{ass:1},~\ref{ass:2} and~\ref{ass:ex}, the process $\hat{Z}_t$ and
any benchmarked portfolio process are supermartingales which are not local martingales.
\label{theorem:ex}
\end{thm}
\begin{proof}
By analogy to the proof of Proposition \ref{proposition:num}, we
start by calculating the SDE for
$\bar{S}_t^\delta/\bar{S}_t^{\delta_*}$, where $\delta_t$ is an
arbitrary admissible trading strategy. According to the product
formula we proceed by calculating separately the components
$d\bar{S}_t^\delta$, $d\left(\frac{1}{\bar{S}_t^{\delta_*}}\right)$
and $d\left[
\bar{S}_t^\delta,\frac{1}{\bar{S}_t^{\delta_*}}\right]$. The first
component can be obtained by adjusting the general formula
in~\eqref{eq:S} to the case in which $d=2$ so that
\[
\begin{split}
d\bar{S}_t^\delta&=\left(\delta_t^1\bar{S}_t^1b_t^{1,1}+\delta_t^2\bar{S}_t^2b_t^{2,1}\right)
\theta_t^1dt+\left(\delta_t^1\bar{S}_t^1b_t^{1,2}+\delta_t^2\bar{S}_t^2b_t^{2,2}\right)\theta_t^2dt\\
&\;\;+\left(\delta_t^1\bar{S}_t^1b_t^{1,1}+\delta_t^2\bar{S}_t^2b_t^{2,1}\right)dW_t
+\left(\delta_t^1\bar{S}_{t-}^1b_t^{1,2}+\delta_t^2\bar{S}_{t-}^2b_t^{2,2}\right)dM_t
\end{split}
\]
The second term comes from applying the It\^o formula to (\ref{eq:gopex}), namely
\beq
\label{eq:sz}
\begin{split}
d\left(\frac{1}{\bar{S}_t^{\delta_*}}\right)&=-\frac{1}{\bar{S}_t^{\delta_*}}\left(\left(\theta_t^1\right)^2
+\psi\left(\theta_t^2-\sqrt{\lambda_t}\right)\right)dt+\frac{1}{\bar{S}_t^{\delta_*}}\left(\theta_t^1\right)^2dt\\
&\quad-\frac{1}{\bar{S}_t^{\delta_*}}\theta_t^1dW_t+\frac{1}{\bar{S}_{t-}^{\delta_*}}\left[\left(1+\psi\frac{1}
{\sqrt{\lambda_t}}\right)^{-1}-1\right]dN_t\\
&=-\frac{1}{\bar{S}_t^{\delta_*}}\psi\left(\theta_t^2-\sqrt{\lambda_t}\right)dt-\frac{1}{\bar{S}_t^{\delta_*}}
\theta_t^1dW_t-\frac{1}{\bar{S}_{t-}^{\delta_*}}\frac{\psi}{\sqrt{\lambda_t}+\psi}dN_t\\
&=-\frac{1}{\bar{S}_t^{\delta_*}}\left(\psi\left(\theta_t^2-\sqrt{\lambda_t}\right)+\frac{\psi\lambda_t}
{\sqrt{\lambda_t}+\psi}\right)dt-\frac{1}{\bar{S}_{t-}^{\delta_*}}\left(\theta_t^1dW_t+\frac{\psi\sqrt{\lambda_t}}
{\sqrt{\lambda_t}+\psi}dM_t\right)
\end{split}
\eeq Note that the drift term in the SDE satisfied by
$\frac{1}{\bar{S}_t^{\delta_*}}$ is strictly negative as
Assumption~\ref{ass:test} does not hold and both $\psi$ and
$\lambda_t$ are positive. This shows that the inverse of the
discounted GOP, namely $\hat{Z}_t$, is a strict supermartingale.
Finally,
\[
d\left[
\bar{S}_t^\delta,\frac{1}{\bar{S}_t^{\delta_*}}\right]=-\left(\delta_t^1\hat{S}_t^1b_t^{1,1}
+\delta_t^2\hat{S}_t^2b_t^{2,1}\right)\theta_t^1dt-\left(\delta_t^1\hat{S}_{t-}^1b_t^{1,2}
+\delta_t^2\hat{S}_{t-}^2b_t^{2,2}\right)\frac{\psi\sqrt{\lambda_t}}{\sqrt{\lambda_t}+\psi}\frac{1}{\lambda_t}dN_t
\]
On the other hand, for the benchmarked portfolios we have
\[
\begin{split}
d\left(\frac{\bar{S}_t^\delta}{\bar{S}_t^{\delta_*}}\right)&=\left(\delta_t^1\hat{S}_t^1b_t^{1,2}
+\delta_t^2\hat{S}_t^2b_t^{2,2}\right)\theta_t^2dt+\left(\delta_t^1\hat{S}_t^1b_t^{1,1}+\delta_t^2\hat{S}_t^2b_t^{2,1}\right)dW_t\\
&\quad+\left(\delta_t^1\hat{S}_{t-}^1b_t^{1,2}+\delta_t^2\hat{S}_{t-}^2b_t^{2,2}\right)dM_t
-\left(\delta_t^1\hat{S}_{t-}^1b_t^{1,2}+\delta_t^2\hat{S}_{t-}^2b_t^{2,2}\right)
\frac{\psi}{\sqrt{\lambda_t}+\psi}\frac{1}{\sqrt{\lambda_t}}dN_t\\
&\quad-\hat{S}_t^{\delta}\left(\psi\left(\theta_t^2-\sqrt{\lambda_t}\right)+\frac{\psi\lambda_t}
{\sqrt{\lambda_t}+\psi}\right)dt-\hat{S}_{t-}^{\delta}\left(\theta_t^1dW_t+\frac{\psi\sqrt{\lambda_t}}
{\sqrt{\lambda_t}+\psi}dM_t\right)\\
&=\left[\left(\delta_t^1\hat{S}_t^1b_t^{1,2}+\delta_t^2\hat{S}_t^2b_t^{2,2}\right)\left(\theta_t^2
-\frac{\psi\sqrt{\lambda_t}}{\sqrt{\lambda_t}+\psi}\right)-\hat{S}_t^{\delta}
\left(\psi\left(\theta_t^2-\sqrt{\lambda_t}\right)+\frac{\psi\lambda_t}{\sqrt{\lambda_t}+\psi}\right)\right]dt\\
&\quad+\left(\delta_t^1\hat{S}_t^1b_t^{1,1}+\delta_t^2\hat{S}_t^2b_t^{2,1}\right)dW_t
+\left(\delta_t^1\hat{S}_{t-}^1b_t^{1,2}+\delta_t^2\hat{S}_{t-}^2b_t^{2,2}\right)dM_t\\
&\quad-\left(\delta_t^1\hat{S}_{t-}^1b_t^{1,2}+\delta_t^2\hat{S}_{t-}^2b_t^{2,2}\right)
\frac{\psi}{\sqrt{\lambda_t}+\psi}\frac{dN_t-\lambda_tdt}{\sqrt{\lambda_t}}-\hat{S}_{t-}^{\delta}
\left(\theta_t^1dW_t+\frac{\psi\sqrt{\lambda_t}}{\sqrt{\lambda_t}+\psi}dM_t\right)
\end{split}
\]
from which, being condition~\eqref{eq:ex} equivalent to
$\delta_t^1\hat{S}_t^1b_t^{1,2}+\delta_t^2\hat{S}_t^2b_t^{2,2}\leq\psi
\hat{S}_t^\delta,$ we note that the drift term is negative, in fact
\[
\begin{split}
&\left(\delta_t^1\hat{S}_t^1b_t^{1,2}+\delta_t^2\hat{S}_t^2b_t^{2,2}\right)\left(\theta_t^2
-\frac{\psi\sqrt{\lambda_t}}{\sqrt{\lambda_t}+\psi}\right)-\hat{S}_t^{\delta}
\left(\psi\left(\theta_t^2-\sqrt{\lambda_t}\right)+\frac{\psi\lambda_t}{\sqrt{\lambda_t}+\psi}\right)\leq\\
&\leq
\hat{S}_t^\delta\left(\psi\left(\theta_t^2-\frac{\psi\sqrt{\lambda_t}}{\sqrt{\lambda_t}
+\psi}\right)-\psi\left(\theta_t^2-\sqrt{\lambda_t}\right)-\frac{\psi\lambda_t}{\sqrt{\lambda_t}+\psi}\right)=0
\end{split}
\]
This is because the factor
$\theta_t^2-\frac{\psi\sqrt{\lambda_t}}{\sqrt{\lambda_t}+\psi}$ is
positive: the function $x\mapsto\frac{\psi x}{x+\psi}$ is always
increasing, so that, since in this section we are violating
Assumption~\ref{ass:test}, we have
\[
\theta_t^2-\frac{\psi\sqrt{\lambda_t}}{\sqrt{\lambda_t}+\psi}\geq\theta_t^2-\frac{\psi\theta_t^2}
{\theta_t^2+\psi}=\frac{\left(\theta_t^2\right)^2}{\theta_t^2+\psi}>0
\]
\end{proof}
The above theorem shows that, if trading is restricted, the investor
maximizing the expected logarithmic utility function (the GOP
maximizes also the expected log-utility) goes as far as the
admissibility constraints, expressed in Assumption~\ref{ass:ex},
allow. Notice that, since $\hat{Z}_t$ is a supermartingale without
being a local martingale, it cannot be the density process of an
ELMM.

For the sake of completeness we now study the property of
$\hat{Z}_t$ when Assumption~\ref{ass:test} holds true but we enforce
Assumption~\ref{ass:ex} as well. We shall see that $\hat{Z}_t$ is
again a supermartingale that is not  a local martingale and that
$S_t^{\delta_*}$ is the numeraire portfolio.
\begin{prop}  Under Assumptions~\ref{ass:1},~\ref{ass:2},~\ref{ass:test} and~\ref{ass:ex}, if
there exists $B\subset [0,T]$ with positive Lebesgue measure such
that for $t\in B$ one has
$\psi<\frac{\theta_t^2}{1-\frac{\theta_t^2}{\sqrt{\lambda_t}}}$,
then any benchmarked portfolio process is a supermartingale which is
not a local martingale.
\end{prop}
\begin{proof} We start by noting that the condition
$\psi<\frac{\theta_t^2}{1-\frac{\theta_t^2}{\sqrt{\lambda_t}}}$
simply requires that Assumption~\ref{ass:ex} imposes a real
constraint on the investor, who would otherwise construct theoptimal
portfolio as if the condition expressed in Assumption~\ref{ass:ex}
was not present. Therefore, as soon as this condition holds, the
discounted GOP dynamics is the same as~\eqref{eq:gopex}, and the
SDEs satisfied by the inverse of the discounted GOP, namely
$\hat{Z}_t$, and by any benchmarked portfolio are the same as those
obtained in the proof of Theorem~\ref{theorem:ex}. The only thing
left to do is to check the negativity of the drift term in each of
the dynamics. For the former, the drift term of $\hat{Z}_t$, we get
from~\eqref{eq:sz}
\[
\begin{split}
-\hat{Z}_t\left(\psi\left(\theta_t^2-\sqrt{\lambda_t}\right)+\frac{\psi\lambda_t}
{\sqrt{\lambda_t}+\psi}\right)&=-\hat{Z}_t\left(\frac{\psi\left(\theta_t^2
-\sqrt{\lambda_t}\right)\left(\sqrt{\lambda_t}+\psi\right)+\psi\lambda_t}{\sqrt{\lambda_t}+\psi}\right)\\
&=-\hat{Z}_t\left(\frac{\psi\left(\psi\left(\theta_t^2-\sqrt{\lambda_t}\right)
+\theta_t^2\sqrt{\lambda_t}\right)}{\sqrt{\lambda_t}+\psi}\right)
\end{split}
\]
which is negative since $\psi>0$ and
$\psi<\frac{\theta_t^2}{1-\frac{\theta_t^2}{\sqrt{\lambda_t}}}$. The
latter follows in the same way as in the proof of
Theorem~\ref{theorem:ex}, noting that
$\theta_t^2-\frac{\psi\sqrt{\lambda_t}}{\psi+\sqrt{\lambda_t}}>0
\quad\Longleftrightarrow \quad
\psi<\frac{\theta_t^2}{1-\frac{\theta_t^2}{\sqrt{\lambda_t}}}$.
\end{proof}

\end{document}